\documentclass[12pt]{article}

\def\colorful{0}

\usepackage{amsthm,amsfonts,amsmath,amssymb,epsfig,color,float,graphicx,verbatim}
\usepackage{multirow}

\newif\ifhyper\IfFileExists{hyperref.sty}{\hypertrue}{\hyperfalse}
\hypertrue
\ifhyper\usepackage{hyperref}\fi

\usepackage{enumitem}
\usepackage[capitalize]{cleveref} 

\makeatletter
\renewcommand{\section}{\@startsection{section}{1}{0pt}{-12pt}{5pt}{\large\bf}}
\renewcommand{\subsection}{\@startsection{subsection}{2}{0pt}{-12pt}{-5pt}{\normalsize\bf}}
\renewcommand{\subsubsection}{\@startsection{subsubsection}{3}{0pt}{-12pt}{-5pt}{\normalsize\bf}}
\makeatother

\usepackage{framed}
\usepackage{nicefrac}

\def\nnewcolor{1}
\ifnum\nnewcolor=1

\fi
\ifnum\nnewcolor=0

\fi

\ifnum\colorful=1
\newcommand{\new}[1]{{\color{red} #1}}

\else
\newcommand{\new}[1]{{#1}}

\fi

\newtheorem{theorem}{Theorem}

\newtheorem{lemma}[theorem]{Lemma}
\newtheorem{proposition}[theorem]{Proposition}
\newtheorem{corollary}[theorem]{Corollary}
\newtheorem{claim}[theorem]{Claim}

\theoremstyle{definition}
\newtheorem{definition}[theorem]{Definition}

\newcommand{\R}{\mathbb{R}}

\newcommand{\E}{\mathbb{E}}

\newcommand{\poly}{\mathrm{poly}}

\newcommand{\ignore}[1]{}

\newcommand{\eps}{\varepsilon}

\newcommand{\abs}[1]{\lvert#1\rvert}

\newcommand{\domain}[1]{[#1]}
\newcommand{\bigabs}[1]{\left\lvert#1\right\rvert}
\newcommand{\size}[1]{\lvert#1\rvert}
\newcommand{\norm}[1]{\lVert#1\rVert}
\newcommand{\floor}[1]{\lfloor#1\rfloor}
\newcommand{\Var}{\mathop{\textnormal{Var}}\nolimits}
\newcommand{\Bin}{\mathop{\textnormal{Binom}}\nolimits}
\newcommand{\Binom}{\mathop{\textnormal{Binom}}\nolimits}
\newcommand{\bin}{\mathop{\textnormal{bin}}\nolimits}
\newcommand{\Poi}{\mathop{\textnormal{Poi}}\nolimits}

\newcommand{\indic}[1]{\textbf 1_{#1}}
\newcommand{\stir}[2]{\left\{\begin{matrix}#1 \\ #2\end{matrix}\right\}}
\newcommand{\lin}{\mathcal L}
\newcommand{\oper}{\overline}

\renewcommand{\eqref}[1]{Eq.~(\ref{#1})}

\newenvironment{algorithm}[1][\  ] %
{ \rm
\begin{tabbing}
....\=.....\=.....\=.....\=.....\=  \+ \kill
} %
{\end{tabbing} }

\floatstyle{ruled}
\newfloat{Algorithm}{H}{alg}
\newfloat{Subroutine}{H}{sub}

{
\begin{minipage}{1.0\linewidth} \begin{algorithm} %
} { \end{algorithm} \end{minipage} }

\title{Optimal Algorithms for Testing Closeness of Discrete Distributions}

\author{Siu-On Chan\thanks{Supported by NSF award
DMS-1106999, DOD ONR grant N000141110140 and NSF award CCF-1118083.}\\
MSR New England\\
{\tt siuon@cs.berkeley.edu}.\\
\and
Ilias Diakonikolas\thanks{Supported in part by a SICSA PECE grant.
Part of this work was done while the author was at UC Berkeley supported by a Simons Postdoctoral Fellowship.}\\
University of Edinburgh\\
{\tt ilias.d@ed.ac.uk}.\\
\and
Gregory Valiant\thanks{The majority of this work was done while the author was at Microsoft Research.}\\
Stanford University\\
{\tt gregory.valiant@gmail.com}.\\
\and
Paul Valiant\\
Brown University \\
{\tt pvaliant@gmail.com}.
}

\begin{document}

\maketitle

\thispagestyle{empty}

\begin{abstract}
We study the question of closeness testing for two discrete distributions.
More precisely, given samples from two distributions $p$ and $q$ over an $n$-element set,
we wish to distinguish whether $p=q$ versus $p$ is at least $\eps$-far from $q$,
in either $\ell_1$ or $\ell_2$ distance.  Batu et al~\cite{BFR+:00, Batu13} gave the first sub-linear time algorithms for these problems, which matched the lower bounds of~\cite{PV11sicomp} up to a logarithmic factor in $n$, and a polynomial factor of $\eps.$

In this work, we present simple (and new) testers for both the $\ell_1$ and $\ell_2$ settings, with sample complexity
that is information-theoretically optimal, to constant factors, both in the dependence on $n$, and the dependence on $\eps$; for the $\ell_1$ testing problem we establish that the sample complexity is $\Theta(\max\{n^{2/3}/\eps^{4/3}, n^{1/2}/\eps^2 \}).$
\end{abstract}

\thispagestyle{empty}
\setcounter{page}{0}

\newpage

\section{Introduction}  \label{sec:intro}

Consider the following natural statistical task: Given independent samples from a pair of unknown distributions $p$, $q$, determine whether
the two distributions are {\em the same} versus significantly different. We focus on the most basic (and well-studied) setting
in which both $p$ and $q$ are discrete distributions supported on a set of size $n$. For a parameter $0<\eps<1$,
we want to distinguish (with probability at least $2/3$, say) between the case that $p=q$ and the case that $p$ and $q$ are {\em $\eps$-far} from each other,
i.e., the $\ell_1$ distance between $p$ and $q$ is at least $\eps$.
We will henceforth refer to this task as the problem of {\em closeness testing} for $p$ and $q$.

We would like to design an algorithm (tester) for this task that uses as few samples as possible
and is computationally efficient (i.e., has running time polynomial in its sample size).
One natural way to solve this problem would be to get sufficiently many samples from $p$, $q$ in order to {\em learn} each distribution
to accuracy $O(\eps)$, and then check closeness of the corresponding hypothesis distributions.
As natural as it may be, this testing-via-learning approach is quite naive and gives suboptimal results.
We note that learning an arbitrary distribution over support of size $n$ to $\ell_1$ distance $\eps$ requires $\Theta(n/\eps^2)$ samples (i.e., there is an upper bound
of $O(n/\eps^2)$ and a matching information-theoretic lower bound of $\Omega(n/\eps^2)$).
One might hope that a better sample size bound could be achieved for the closeness testing problem, since this task is, in some sense, more specific than the general task of learning.
Indeed, this is known to be the case: previous work~\cite{BFR+:00} gave a tester for this problem with sample complexity {\em sub-linear} in $n$
and polynomial in $1/\eps$.

Despite its long history in both statistics and computer science, the sample complexity of this basic task has not been resolved to date.
While the dependence on $n$ in the previous bound~\cite{BFR+:00} was subsequently shown~\cite{PValiant:08, PV11sicomp}
to be tight to within logarithmic factors of $n$, there was a polynomial gap between the upper and lower bounds in the dependence on $\eps$.
Due to its fundamental nature, we believe it is of interest from a theoretical standpoint to obtain an {\em optimal} sample (and time) algorithm for the problem.
From a practical perspective, we note that in an era of ``big data'' it is critical to use data efficiently. In particular,
in such a context, even modest asymptotic differences in the sample complexity can play a big role.

\smallskip

{\em In this paper, we resolve the complexity of the closeness testing problem, up to a constant factor, by designing a sample-optimal algorithm (tester)
for it whose running time is linear in the sample size.} Our tester has a different structure from the one in~\cite{BFR+:00} and is also much simpler.
We also study the  closeness testing problem with respect to the $\ell_2$ distance metric between distributions.
This problem, interesting in its own right, has been explicitly studied in previous work~\cite{GR00, BFR+:00}.


As our second contribution, we design a similarly optimal algorithm for closeness testing in the $\ell_2$ norm.  In this $\ell_2$ setting, we show that the \emph{same} sample complexity allows one to ``robustly'' test closeness; namely, the same sample complexity allows one to distinguish the case that $||p-q||_2 \le \eps$ from the case that $||p-q||_2 \ge 2 \eps.$   This correspondence between the robust and non-robust closeness testing in the $\ell_2$ setting does not hold for the $\ell_1$ setting: the lower bounds of~\cite{ValiantValiant:11focs} show that robust  $\ell_1$ testing for distributions of support size $n$ requires $\Theta(\frac{n}{\log n})$ samples (for constant $\eps$), as opposed to the $\Theta(n^{2/3})$ for the non-robust testing problem. One may alternately consider ``robust" closeness testing under the $\ell_2$ norm as essentially the problem of \emph{estimating} the $\ell_2$ distance, and the results of Proposition~\ref{prop:ggg} are presented from this perspective.

Algorithmic ideas developed for the closeness testing problem have typically been useful
for related testing questions, including the independence of bivariate distributions (see e.g.~\cite{BFFKRW:01, BKR:04}). It is plausible that our techniques may be used to obtain similarly optimal algorithms
for these problems, but we have not pursued this direction.

\smallskip

Before we formally state our results, we start by providing some background in the area of distribution property testing.

\medskip

\noindent {\bf Related Work.} Estimating properties of distributions using samples is a classical topic in statistics that has received considerable attention
in the theoretical CS community during the past decade; see~\cite{GR00, BFR+:00, BFFKRW:01, Batu01, BDKR:02, BKR:04,  Paninski:08, PValiant:08, Onak09, PV11sicomp, ValiantValiant:11,ValiantValiant:11focs, DDSVV13, Rub12, BaNNR11, DJOP11, LRR11, ILR12, AIOR11} for a sample of works
and~\cite{Rub12} for a recent survey on the topic. In addition to closeness testing, various properties of distributions
have been considered, including independence~\cite{BFFKRW:01, Onak09}, entropy~\cite{ BDKR:02}, and the more general class of ``symmetric'' properties~\cite{PValiant:08,ValiantValiant:11,ValiantValiant:11focs},
monotonicity~\cite{BKR:04}, etc.

One of the first theoretical CS papers that explicitly studied such questions is the work of Batu et al~\cite{BFR+:00} (see~\cite{Batu13} for the journal version). In this work, the authors formally pose the closeness testing problem
and give a tester for the problem with sub-linear sample complexity. In particular, the sample complexity of their algorithm under the $\ell_1$ norm
is $O(\frac{n^{2/3} \log n}{\eps^{8/3}})$.
A related (easier) problem is that of {\em uniformity testing}, i.e., distinguishing between the case that an unknown distribution $p$ (accessible via samples) is uniform versus
$\eps$-far from uniform. Goldreich and Ron~\cite{GR00}, motivated by a connection to testing expansion in graphs, obtained a uniformity tester using $O(\sqrt{n}/\eps^4)$ samples. Subsequently, Paninski gave the tight bound of $\Theta(\sqrt{n}/\eps^2)$~\cite{Paninski:08}. (Similar results are obtained for both testing problems under the $\ell_2$ norm.)

\medskip

\noindent{\bf Notation.} We write $[n]$ to denote the set $\{1, \ldots, n\}$. We consider discrete probability distributions over $[n]$, which are functions
$p: [n] \rightarrow [0,1]$ such that $\sum_{i=1}^n p_i =1.$ We will typically use the notation $p_i$ to denote the probability of element
$i$ in distribution $p$. The $\ell_1$ (resp. $\ell_2$) norm of a distribution is identified with the $\ell_1$ (resp. $\ell_2$) norm of the corresponding $n$-vector, i.e.,
$\|p\|_1 = \sum_{i=1}^n |p_i|$ and $\|p\|_2 = \sqrt{\sum_{i=1}^n p^2_i}$. The $\ell_1$ (resp. $\ell_2$) distance between distributions $p$ and $q$
is defined as the  the $\ell_1$ (resp. $\ell_2$) norm of the vector of their difference, i.e., $\|p-q\|_1 = \sum_{i=1}^n |p_i -q_i|$ and
$\|p-q\|_2 = \sqrt{\sum_{i=1}^n (p_i-q_i)^2}$.   For $\lambda \ge 0$, we denote by $\Poi(\lambda)$ the Poisson distribution with parameter
$\lambda.$ 

\medskip

\noindent {\bf Our Results.} Our main result is an optimal algorithm for the $\ell_1$-closeness testing problem:
\begin{theorem} \label{thm:main}
Given $\eps>0$ and sample access to distributions $p$ and $q$ over $[n]$, there is an algorithm which uses $O(\max\{n^{2/3}/\eps^{4/3}, n^{1/2}/\eps^2 \})$ samples,
runs in time linear in its sample size and with probability at least $2/3$ distinguishes whether $p=q$ versus $\|p-q\|_1 \ge \eps$.  Additionally, $\Omega(\max\{n^{2/3}/\eps^{4/3}, n^{1/2}/\eps^2 \})$ samples are information-theoretically necessary.
\end{theorem}

The lower bound is obtained by leveraging the techniques of~\cite{PV11sicomp} to show that $\Omega(n^{2/3}/\eps^{4/3})$ is a lower bound, as
long as $\eps = \Omega(n^{-1/4})$ (see Section~\ref{sec:lb} for the proof). On the other hand,
the sample complexity of $\ell_1$-closeness testing is bounded from below by
the sample complexity of uniformity testing (for all values of $n$ and $\eps>0$), since knowing that one distribution is exactly the uniform distribution can only make the testing problem easier.

Hence, by the result of Paninski~\cite{Paninski:08}, it follows that $\Omega(\sqrt{n}/\eps^2)$ is also a lower bound.
The tight lower bound of $\Omega(\max\{n^{2/3}/\eps^{4/3}, n^{1/2}/\eps^2 \})$ follows from the fact that the two functions intersect for $\eps = \Theta(n^{-1/4})$.
Hence, our algorithm of Theorem~\ref{thm:main} is optimal (up to constant factors) for all $\eps>0$.

Our second result is an algorithm for ``robustly'' testing the closeness of a pair of distributions with respect to $\ell_2$ distance, which is also information theoretically optimal for all parameters, to constant factors. The parameter $b$ in the following theorem upper-bounds the $\ell_2$ norm-squared of each distribution, which allows the theorem to be more finely tuned to the cases when testing should be easier or harder.
\begin{theorem} \label{thm:l2}
For two distributions $p,q,$ over $[n]$ with $b \ge ||p||_2^2,||q||_2^2,$ there is an algorithm which distinguishes the case that $||p-q||_2 \le \eps$ from the case that $||p-q||_2 \ge 2\eps$ when given $O(\sqrt{b}/\eps^2)$ samples from $p$ and $q$  with probability at least $2/3$. This is information theoretically optimal, as distinguishing the case that $p=q$ from the case that $||p-q||_2 > 2\eps$ requires $\Omega(\sqrt{b}/\eps^2)$ samples.
\end{theorem}

We note that both the upper and lower bounds of the above theorem continue to hold if $b$ is defined to be an upper bound on $||p||_{\infty},||q||_{\infty}$; the upper bound trivially holds because, for all $p$, $||p||_2^2 \le \max_i p_i$, and the lower bound holds because the specific lower bound instance we construct consists of nearly uniform distributions for which $||p||_2^2 \ge \max_i p_i/2.$ See Proposition~\ref{prop:ggg} and the discussion following it for analysis of our algorithm as an \emph{estimator} for $\ell_2$ distance.

\medskip

\noindent {\bf The $\ell_2\rightarrow\ell_1$ testing approach.}
Recall that the $\ell_1$ closeness tester in~\cite{BFR+:00} proceeds in two steps:
In the first step, it ``filters'' the elements of $p$ and $q$ that are ``$b$-heavy'', i.e.,
have probability mass at least $b$ -- for an appropriate value of $b$.
(This step essentially amounts to {\em learning} the heavy
parts of $p$ and $q$.) In the second step, it uses an $\ell_2$ closeness
tester applied to the ``light'' parts of $p$ and $q$.
The $\ell_2$ tester used in~\cite{BFR+:00}  is a generalization of a tester
proposed in~\cite{GR00}.

\new{Using such a two step approach, Theorem~\ref{thm:l2} can be used as a black-box to obtain an $\ell_1$  closeness tester with sample complexity $O(n^{2/3} \log n/ \eps^2)$.  This can further be improved to $O(n^{2/3}/ \eps^2)$ by improving the ``filtering'' algorithm of~\cite{BFR+:00};  in Appendix~\ref{ssec:improved-l1} we describe an optimal ``filtering'' algorithm, which might be applicable in other settings.
Curiously, since the sample complexity of both the improved filtering algorithm, and the $\ell_2$ tester are optimal, the corresponding sample complexity of $O(n^{2/3}/ \eps^2)$ for the $\ell_1$ testing problem seems to be the best that could possibly be achieved via this reduction-based approach.  This suggests that, in some sense, our novel (and more direct) approach underlying Theorem~\ref{thm:main} is necessary to achieve the optimal $\eps$-dependence for the $\ell_1$ testing problem.
}

\bigskip

{
\noindent {\bf Structure of the paper.} In  Section~\ref{sec:L1} we present our $\ell_1$ tester, and in Section~\ref{sec:L2} we present our $\ell_2$ tester.
In Section~\ref{sec:lb} we prove the information theoretic lower bounds, establishing the optimality of both testers.
The details of the reduction--based (though suboptimal) $\ell_1$ closeness tester can be found in the appendix.
}

\bigskip

\noindent {\bf Remark.} Throughout our technical sections, we employ the standard ``Poissonization'' approach: namely, we assume that, rather than drawing $k$ independent samples from a distribution, we first select $k'$ from $\Poi(k)$, and then draw $k'$ samples.  This Poissonization makes the number of times different elements occur in the sample independent, simplifying the analysis.  As $\Poi(k)$ is tightly concentrated about $k$, we can carry out this Poissonization trick without loss of generality at the expense of only subconstant factors in the sample complexity.

\section{Closeness testing in $\ell_1$ norm} \label{sec:L1}

We begin by describing our $\ell_1$ closeness testing algorithm:
\medskip

\fbox{\parbox{6in}{
Input: A constant $C$ and $m$ samples from distributions $p,q$, with $X_i,Y_i$ denoting the number of occurrences of the $i$th domain elements in the samples from $p$ and $q$, respectively.\\
\begin{enumerate}
  \item Define \begin{equation}\label{eq-alg}Z=\sum_i \frac{(X_i-Y_i)^2-X_i-Y_i}{X_i+Y_i}.\end{equation}
  \item If $Z \leq C \cdot \sqrt{m}$ then output EQUAL, else output DIFFERENT.
\end{enumerate}
}}
\medskip

The following proposition characterizes the performance of the above tester, establishing the algorithmic portion of Theorem~\ref{thm:main}.

\begin{proposition}\label{prop:l1}
There exist absolute constants $C,C'$ such that the above algorithm, on input $C$ and a set of $\Poi(m)$ samples drawn from two distributions, $p,q$, supported on $[n]$,  will correctly distinguish the case that $p=q$ from the case that $||p-q||_1 \ge \eps$, with probability at least $2/3$ provided that $m \ge C' \max\{n^{2/3}/\eps^{4/3},n^{1/2}/\eps^2\}.$
\end{proposition}

We will show that the error probability of the above algorithm is $O(\frac{1}{C^2})$, hence for a suitable constant $C$ the tester succeeds with probability $\frac{2}{3}$. (Repeating the tester and taking the majority answer results in an exponential decrease in the error probability.)

The form of the right hand side of \eqref{eq-alg} is rather similar to our $\ell_2$ distance tester (given in the next section), though the difference in normalization is crucial. However, though we do not prove corresponding theorems here, the right hand side of \eqref{eq-alg} can have a variety of related forms while yielding similar results, with possibly improved constants. For example, one could use $\sum_i |X_i-Y_i|-f(X_i+Y_i)$, where $f(j)$ is the expected deviation from $j/2$ heads in $j$ fair coin flips, which is ${j-1\choose\lfloor(j-1)/2\rfloor}\frac{j}{2^j}$.

\medskip

In the remainder of this section we prove Proposition~\ref{prop:l1}.  First, letting $p_i,q_i$ respectively denote the probabilities of the $i$th elements in each distribution, note that
if $p_i=q_i$ then the expectation of the sum in \eqref{eq-alg} is 0, as can be seen by conditioning the summand for each $i$ on the value of $X_i+Y_i$: subject to this, $X_i,Y_i$ can be seen as the number of heads and tails respectively found in $X_i+Y_i$ fair coin flips, and $\E[(X_i-Y_i)^2]$ is 4 times the variance of $X_i$ alone, which is a quarter of the number of coin flips, and thus the expression in total has expectation 0.

When $p \neq q$, we use the following lemma to bound from below the expected
value of our estimator in terms of $\norm{p-q}_1$.

\begin{lemma}\label{L1-prop1}
For $Z$ as defined in Equation~\ref{eq-alg}, $\E[Z] \ge \frac{m^2}{4n+2m}
\norm{p-q}_1^2.$
\end{lemma}

\begin{proof}
Conditioned on $X_i+Y_i=j$, for some $j$, we have that $X_i$ is distributed as the number of heads in the distribution $\Binom(j,\frac{p_i}{p_i+q_i})$. For the distribution $\Binom(j,\alpha)$, the expected value of the square of the difference between the number of heads and tails can be easily seen to be
$4j^2(\frac{1}{2}-\alpha)^2+4j\alpha(1-\alpha);$ we subtract $j$ from this because of the $-X_i-Y_i$ term in the numerator of \eqref{eq-alg} to yield $4(j^2-j)(\frac{1}{2}-\alpha)^2$, and divide by $j$ because of the denominator of \eqref{eq-alg} to yield $4(j-1)(\frac{1}{2}-\alpha)^2$. Plugging in $\alpha=\frac{p_i}{p_i+q_i}$ yields $(j-1)(\frac{p_i-q_i}{p_i+q_i})^2$. Thus the expected value of the summand of \eqref{eq-alg}, for a given $i$, conditioned on $X_i+Y_i=j$ is this last expression, if $j\neq 0$, and 0 otherwise. Thus the expected value of the summand across all $j$, since $\E[j]=m(p_i+q_i)$, equals $$m\frac{(p_i-q_i)^2}{p_i+q_i}-(1-e^{-m(p_i+q_i)})(\frac{p_i-q_i}{p_i+q_i})^2,$$ where we have used the fact that $\Pr[X_i+Y_i=0]=e^{-m(p_i+q_i)}$. Gathering terms, we conclude that the expectation of each term of \eqref{eq-alg} equals \begin{equation}\label{eq-alg-ex}\frac{(p_i-q_i)^2}{p_i+q_i}m\left(1-\frac{1-e^{-m(p_i+q_i)}}{m(p_i+q_i)}\right)\end{equation}

Defining the function $g(\alpha)=\alpha\left/\left(1-\frac{1-e^{-\alpha}}{\alpha}\right)\right.$, this expression becomes $m^2\frac{(p_i-q_i)^2}{g(m(p_i+q_i))}$, and we bound its sum via Cauchy-Schwarz as \[m^2\left(\sum_i\frac{(p_i-q_i)^2}{g(m(p_i+q_i))}\right)\left(\sum_i g(m(p_i+q_i))\right)\geq m^2\left(\sum_i |p_i-q_i|\right)^2\]

It is straightforward to bound $g(\alpha)\leq 2+\alpha$, leading to $\sum_i g(m(p_i+q_i))\leq 4n+2m$, since the support of each distribution is at most $n$ and each has total probability mass 1. Thus the expected value of the left hand side of \eqref{eq-alg} is at least $\frac{m^2}{4n+2m}\left(\sum_i |p_i-q_i|\right)^2$.
\end{proof}

We now bound the variance of the $i$th term of $Z$.

\begin{lemma}\label{L1-prop2}
For $Z$ as defined in Equation \eqref{eq-alg}, $\Var[Z] \le 2\min\{n,m\} +
\sum_i 5m \frac{(p_i-q_i)^2}{p_i+q_i}$.
\end{lemma}

\begin{proof}
To bound the variance of the $i$th term of $Z$, we will split this variance calculation into two parts: the variance conditioned on $X_i+Y_i=j$, and the component of the variance due to the variation in $j$.  Letting $$f(X_i,Y_i)=\frac{(X_i-Y_i)^2-X_i-Y_i}{X_i+Y_i},$$ we have that $$\Var[f(X,Y)] \le \max_j\left( \Var[f(X,Y)|X+Y=j]\right)+\Var[\E[f(X,Y)|X+Y=j]].$$  We now bound the first term; since $(X_i - Y_i)^2 = (j - 2Y_i)^2$, and $Y_i$ is distributed as
$\Binom(j; \frac{q_i}{p_i+q_i})$ where for convenience we let  $\alpha=\frac{q_i}{p_i+q_i}$ we can compute the variance of $(j -2Y_i)^2$ from standard expressions for the moments of the Binomial distribution as $$\Var[(j - 2Y_i)^2] =16j(j -1)\alpha(1 -\alpha)\left( (j-\frac{3}{2})(1-2\alpha)^2+\frac{1}{2}\right).$$ We bound this expression, since $\alpha(1-\alpha) \le \frac{1}{4}$ and $j-\frac{3}{2} < j-1 < j$ as $j^2(2+4j(1-2\alpha)^2)$.  Because the denominator of the $i$th term of \eqref{eq-alg} is $X_i + Y_i = j$, we must divide this by $j^2$, make it $0$ when $j = 0$, and take its expectation as $j$ is distributed as $\Poi(m(p_i + q_i))$, yielding: $$\Var[f(X_i,Y_i)|X_i+Y_i=j] \le 2(1 - e^{-m(p_i+q_i)}) + 4m\frac{(p_i-q_i)^2}{p_i+q_i}.$$

We now consider the second component of the variance---the contribution to the variance due to the variation in the sum $X_i + Y_i$. Since for fixed $j$, as noted above, we have $Y_i$ distributed as $\Binom(j; \frac{q_i}{p_i+q_i}),$ where for convenience we let $\alpha = \frac{q_i}{p_i+q_i},$ we have $$\E[(X_i - Y_i)^2] = \E[j^2-4jY_i+4Y_i^2]=j^2-4j^2\alpha+4(j\alpha-j\alpha^2+j^2\alpha^2)=j^2(1-2\alpha)^2+4j\alpha(1-\alpha).$$ As in \eqref{eq-alg}, we finally subtract $j$ and divide by $j$ to yield $(j - 1)(1 - 2 \alpha)^2$, except with a value of $0$ when $j = 0$ by definition; however, note that replacing the value at $j=0$ with $0$ can only lower the variance. Since  the sum $j=X_i+Y_i$ is drawn from a Poisson distribution with parameter $m(p_i + q_i)$, we thus have: $$\Var \left[ \E[f(X_i,Y_i)|X_i+Y_i=j] \right] \le m(p_i + q_i)(1 - 2\alpha)^4 \le  m(p_i + q_i)(1 - 2\alpha)^2 = m\frac{(p_i-q_i)^2}{p_i+q_i}.$$

Summing the final expressions of the previous two paragraphs yields a bound on the variance of the $i$th term of \eqref{eq-alg} of $$2(1 - e^{-m(p_i+q_i)}) + 5m\frac{(p_i-q_i)^2}{p_i+q_i}.$$ We note that since $1 - e^{-m(p_i+q_i)}$ is bounded
by both $1$ and $m(p_i +q_i)$, the sum of the first part is bounded as $$\sum_i 2(1- e^{-m(p_i+q_i)}) \le 2 \min\{n,m\}.$$
This completes the proof.
\end{proof}

We now complete our proof of Proposition~\ref{prop:l1}, establishing the upper bound of Theorem~\ref{thm:main}.
\begin{proof}[Proof of Proposition~\ref{prop:l1}]
With a view towards applying Chebyshev's inequality, we compare the square of the expectation
of $Z$ to its variance. From Lemma~\ref{L1-prop1}, the expectation equals $$\left( \sum_i \frac{(p_i-q_i)^2}{p_i+q_i} m \left( 1-\frac{1-e^{-m(p_i+q_i)}}{m(p_i+q_i)}\right)\right)^2,$$ which we showed is at least $\frac{m^2}{4n+2m}
\norm{p-q}_1^2$; from Lemma~\ref{L1-prop2}, the variance is at most $$2\min\{n,m\} + \sum_i 5m\frac{(p_i-q_i)^2}{p_i+q_i}.$$

We consider the second part of the variance expression. It is clearly bounded by $10m$, so when $m < n$ the first expression dominates. Otherwise, assume that $m \ge n$. 
Consider the case when our bound on the expectation, $\frac{m^2}{4n+2m}\norm{p-q}_1^2,$ is at least $2$, namely that $m = \Omega(\norm{p-q}_1^{-2}).$ Thus, with a view towards applying Chebyshev's inequality, we can bound the square of the expectation by: $$\left(\sum_i \frac{(p_i-q_i)^2}{p_i+q_i} m \left( 1-\frac{1-e^{-m(p_i+q_i)}}{m(p_i+q_i)}\right)\right)^2 \ge \sum_i \frac{(p_i-q_i)^2}{p_i+q_i} m \left( 1-\frac{1-e^{-m(p_i+q_i)}}{m(p_i+q_i)}\right)\cdot 2.$$

For those $i$ for which the multiplier $\left( 1-\frac{1-e^{-m(p_i+q_i)}}{m(p_i+q_i)}\right) \cdot 2$ is greater than $1$, we have that the $i$th term here is greater than the $i$th term of the expression for the variance, $\sum_i \frac{(p_i-q_i)^2}{p_i+q_i}m;$ otherwise, we have $1 - \frac{1-e^{-m(p_i+q_i)}}{m(p_i+q_i)} \le \frac{1}{2}$ which implies $m(p_i + q_i) \le 2,$ and thus the sum of the remaining terms is bounded by $2n$, which is dominated by the first expression in the variance, $2 \min\{n, m\}$ in the
case under consideration, where $m \ge n$. Thus we need only compare the square of the expectation, which is at least $\frac{m^2}{4n+2m}\norm{p-q}_1^2=\frac{m^2}{O(\max\{n,m\})}\norm{p-q}_1^2,$ to $O(\min\{n,m\}),$  yielding, when $m < n$ a bound
$m = \Omega(n^{2/3}/\norm{p-q}_1^{4/3}),$ and when $m \ge n$ a bound $m = \Omega(n^{1/2}/\norm{p-q}_1^2);$  note that in the latter case, this implies $m = \Omega(\norm{p-q}_1^{-2}),$  which we needed in the derivation above.
\end{proof}

\section{Robust $\ell_2$ testing} \label{sec:L2}

In this section, we give an optimal algorithm for robust closeness testing of distributions with respect to $\ell_2$ distance.  For distributions $p$ and $q$ over $[n]$ with $\ell_2^2$ norm at most $b$ (i.e., $\sum_i p_i^2 \le b,$ and $\sum_i q_i^2 \le b$), the algorithm when given $O(\sqrt{b}/\eps^2)$ samples will distinguish the case that $||p-q||_2 \le \eps$ from the case that $||p-q|| \ge 2\eps,$ with high probability.    Since $||p||_2^2 \le \max_i p_i,$ this sample complexity is also bounded by the corresponding expression with $b$ replaced by a bound on the maximum probability of an element of $p$ or $q$.   As we show in Section~\ref{sec:lb}, this sample complexity is optimal even for the easier testing problem of distinguishing the case that the $\ell_2$ distance is $0$ versus at least $\eps$.

Our algorithm is a very natural linear estimator and is similar to the $\ell_2$ tester of~\cite{BFR+:00}.\\

\medskip
\fbox{\parbox{6in}{
Input: $m$ samples from distributions $p,q$, with $X_i,Y_i$ denoting the number of occurrences of the $i$th domain elements in the samples from $p$ and $q$, respectively.\\
Output: an estimate of $||p-q||_2.$
\begin{enumerate}
  \item Define $Z=\sum_{i}(X_i-Y_i)^2-X_i-Y_i.$
  \item Return $\frac{\sqrt{Z}}{m}$.
\end{enumerate}
}}
\medskip

The following proposition characterizes the performance of the above estimator, establishing the algorithmic portion of Theorem~\ref{thm:l2} from the observation that $||p-q||_4^2 \le ||p-q||_2^2.$

\begin{proposition}\label{prop:ggg}
There exists an absolute constant $c$ such that the above estimator, when given $\Poi(m)$ samples drawn from two distributions, $p,q$ will, with probability at least $3/4$, output an estimate of $||p-q||_2$ that is accurate to within $\pm \eps$ provided that $m \ge c\left(\frac{\sqrt{b}}{\eps^2}+ \frac{\sqrt{b}||p-q||_4^2}{\eps^4}\right),$ where $b$ is an upper bound on $||p||_2^2, ||q||_2^2$.
\end{proposition}

\begin{proof}
Letting $X_i,Y_i$ denote the number of occurrences of the $i$th domain elements in the samples from $p$ and $q$, respectively. Define $Z_1=(X_i-Y_i)^2-X_i-Y_i.$   Since $X_i$ is distributed as $\Poi(m\cdot p_i),$  $\E[Z_i] = m^2\cdot \abs{p_i - q_i}^2$, hence $Z$ is an unbiased estimator for $m^2||p-q||_2^2.$

We compute the variance of $Z_i$ via a straightforward calculation involving standard expressions for the moments of a Poisson distribution\footnote{This calculation can be performed in Mathematica, for example, via the expression {\texttt Variance[TransformedDistribution[(X - Y){\textasciicircum}2 - X -   Y, \{X \textbackslash[Distributed] PoissonDistribution[m p],   Y \textbackslash[Distributed] PoissonDistribution[m q]\}]]}}
: $\Var[Z_i]= 4(p_i-q_i)^2(p_i+q_i)m^3 + 2(p_i + q_i)^2m^2.$

Hence  $$\Var[Z]  = \sum_i \Var[Z_i] = \sum_i \left(4m^3(p_i-q_i)^2 (p_i+q_i)  + 2m^2(p_i + q_i)^2\right).$$ By Cauchy-Schwarz, and since $\sum_i (p_i+q_i)^2 \le 4b$, we have $$\sum_i (p_i-q_i)^2(p_i+q_i) \le \sqrt{\sum_i(p_1-q_i)^4 \sum_i (p_i+q_i)^2} \le 2 ||p-q||_4^2 \sqrt{b}.$$ Hence $$\Var[Z] \le 8m^3 \sqrt{b} ||p-q||_4^2 + 8 m^2 b.$$

By Chebyshev's inequality, the returned estimate of $||p-q||_2$ will be accurate to within $\pm \eps$ with probability at least $3/4$ provided $\eps^2 m^2 \ge 2 \sqrt{8m^3 \sqrt{b} ||p-q||_4^2 + 8 m^2 b},$ which holds whenever $$m \ge 6\frac{\sqrt{b}}{\eps^2}+ 32\frac{\sqrt{b}||p-q||_4^2}{\eps^4},$$ since $m\geq x+y$ implies $m^2\geq mx+y^2$, for any $x,y\geq 0$.
\end{proof}

A slightly different kind of result is obtained if we parameterize by $B=\max\{\max_i p_i,\max_i q_i\}$ instead of $b$---where we note that $B\geq b$. We can replace the Cauchy Schwarz inequality of the proof above with $\sum_i (p_i-q_i)^2(p_i+q_i) \le 2B \sum_i (p_i-q_i)^2=2B ||p-q||_2^2$, yielding, analogously to above, that the tester is accurate to $\pm\epsilon$ when given $c(\frac{\sqrt{B}}{\eps^2}+\frac{B||p-q||_2^2}{\eps^4})$ samples. This matches the lower-bound of $\Omega(\frac{\sqrt{B}}{\eps^2})$ provided the second term is not much larger than the first, namely when $\frac{||p-q||_2}{\eps}=O(B^{-1/2})$. Thus our algorithm approximates $\ell_2$ distance to within $\epsilon$ using the optimal number of samples, provided the $\ell_2$ distance is not a $B^{-1/2}$ factor greater than $\epsilon$. For greater distances, we have not shown optimality.

\bigskip

\new{
\noindent {\bf An $O(n^{2/3}/\eps^2)$ $\ell_1$-tester.}
As noted in the introduction, Theorem~\ref{thm:l2} combined with the two step approach of~\cite{BFR+:00}, immediately leads to an $\ell_1$ tester for distinguishing the case that $p=q$ from $||p-q||_1 \ge \eps$
with sample complexity  $O(n^{2/3} \log n/\eps^2)$. One can use Theorem~\ref{thm:l2}
to obtain an $\ell_1$ tester with sample complexity $O(n^{2/3}/\eps^2)$ -- i.e., saving a factor of $\log n$ in the sample complexity.   While this does not match the $O(\max\{n^{2/3}/\eps^{4/3}, n^{1/2}/\eps^2 \})$ performance of the $\ell_1$ tester described in Section~\ref{sec:L1}, the ideas used to remove the $\log n$ factor might be applicable to other problems, and we give the details in Appendix~\ref{ssec:improved-l1}.

}

\section{Lower bounds} \label{sec:lb}
In this section, we present our lower bounds for closeness testing under $\ell_1$ and $\ell_2$ norms.
We derive the results of this section as applications of the machinery developed in~\cite{PV11sicomp} and~\cite{VV13}.

The lower bounds for $\ell_1$ testing require the following definition:
\begin{definition}
  The $(k,k)$-based moments $m(r,s)$ of a distribution pair $(p,q)$ are
  $k^{r+s}\sum_{i=1}^n p_i^r q_i^s$.
\end{definition}

\begin{theorem}[\cite{PV11sicomp}, Theorem 4.18] \label{thm:wt}
  If distributions $p_1^+$,$p_2^+$, $p_1^-$, $p_2^-$ have probabilities at most
  $1/1000k$, and their $(k,k)$-based moments $m^+,m^-$ satisfy
  \[ \sum_{r+s\geq 2} \frac{\abs{m^+(r,s)-m^-(r,s)}}{\floor{\frac
        r2}!\floor{\frac s2}!\sqrt{1+\max\{m^+(r,s),m^-(r,s)\}}} < \frac 1{360}
    , \]
  then the distribution pair $(p_1^+,p_2^+)$ cannot be distinguished with probability $13/24$ from
  $(p_1^-,p_2^-)$ by tester that takes $Poi(k)$ samples from each distribution.
\end{theorem}

The optimality of our $\ell_1$  tester, establishing the lower bound of Theorem~\ref{thm:main},  follows from the following proposition together with the lower bound of $\sqrt{n}/\eps^2$ for testing uniformity given in~\cite{Paninski:08}.

\begin{proposition}
  If $\eps \geq 4^{3/4} n^{-1/4}$, then $\Omega(n^{2/3}\eps^{-4/3})$ samples
  are needed for $0$-vs-$\eps$ closeness testing under the $\ell_1$ norm.
\end{proposition}

\begin{proof}
  Let $b = \eps^{4/3}/n^{2/3}$ and $a = 4/n$, where the restriction on $\eps$ yields that $b\ge a$.
  Let $p$ and $q$ be the distributions
  \[ p = b \indic A + \eps a \indic B \qquad q = b \indic A + \eps a \indic C
  \]
  where $A$, $B$ and $C$ are disjoint subsets of size $(1-\eps)/b$, $1/a$ and
  $1/a$---where the notation $\indic A$ denotes the indicator function that is 1 on the set $A$.
  Then $\norm{p-q}_1 = 2\eps$.
  Let $k = cn^{2/3}\eps^{-4/3}$ for a small enough constant $0 < c < 1$, so
  that $\norm p_\infty = \norm q_\infty = b \leq \frac 1{1000k}$, since $b\geq a$.

  Let $(p_1^+,p_2^+) = (p,p)$ and $(p_1^-,p_2^-) = (p,q)$, so that they have
  $(k,k)$-based moments
  \[ m^+(r,s) = k^t(1-\eps)b^{t-1} +
    k^t\eps^ta^{t-1} \qquad m^-(r,s) = k^t(1-\eps)b^{t-1},  \]
  for $r,s\geq 1$, where $t = r+s$.
  We have the inequality
  \[ \frac {\abs{m^+(r,s)-m^-(r,s)}}{\sqrt{1+\max\{m^+(r,s),m^-(r,s)\}}} \leq
    \frac{k^t\eps^ta^{t-1}}{\sqrt{k^t(1-\eps)b^{t-1}}} . \]
  For $t\geq 2$, it is at most $k^{t/2}\eps^ta^{t-1}/b^{(t-1)/2} \leq c^{t/2}4^{(2t-1)/3}$
  (where we used that $\eps \geq 4/n$). Further, when one of $r$ or $s$ is 0, the moments are equal, since $p$ and $q$ are permutations of each other, yielding a contribution of 0 to the expression of Theorem~\ref{thm:wt}.
  Thus the expression in Theorem~\ref{thm:wt} is bounded by $O(c)$ as the sum of a geometric series (in two dimensions), and thus the distribution pairs $(p,p)$ and
  $(p,q)$ are indistinguishable by Theorem~\ref{thm:wt}.
\end{proof}

The optimality of our $\ell_2$ tester will follow from the following result from~\cite{VV13}:

\begin{theorem}[\cite{VV13}, Theorem 3] \label{thm:hd}
  Given a distribution $p$, and associated values $\eps_i \in [0,p_i],$ define the distribution over distributions, $Q_{p,\eps}$ by the following process: for each domain element $i$, randomly choose $q_i=p_i \pm \eps_i,$ and then normalize $q$ to be a distribution.  There exists a constant $c$ such that it takes at least $c \left(\sum_i \frac{\eps_i^4}{p_i^2} \right)^{-1/2}$ samples to distinguish $p$ from a sample drawn from a random element of $Q_{p,\eps}$ with success probability at least $2/3.$
\end{theorem}

The following proposition establishes the lower bound of Theorem~\ref{thm:l2}, showing the optimality of our $\ell_2$ tester.  Note that if $\max_i p_i\le b$ and $\max_i q_i \le b,$ then $||p-q||_2 \le \sqrt{2b},$ hence the testing problem is trivial unless $\eps \le \sqrt{2b}.$

\begin{proposition}
For any $b \in [0,1],$ and $\eps \le \sqrt{b},$  there exists a distribution $p_b$ and a family of distributions $T_{p,\eps}$ such that for a $q \leftarrow T$ chosen uniformly at random, the following hold:
\begin{itemize}
\item $||p||_2^2 \in [b/2,b]$ and $\max_i p_i \in [b/2,b]$ and with probability at least $1-o(1),$  $||q||_2^2 \in [b/2,b]$ and $\max_i q_i \in [b/2,b].$
\item With probability at least $1-o(1)$, $||p-q||_2 \ge \eps/2.$
\item No algorithm can distinguish a set of $k=c \frac{\sqrt{b}}{\eps^2}$ samples from $q$ from a set drawn from $p$ with probability of success greater than $3/4$, hence no algorithm can distinguish sets of $k$ samples drawn from the pair $(p,p)$ versus drawn from $(p,q)$ with this probability.
\end{itemize}
\end{proposition}
\begin{proof}
Assume for the sake of clarity that $1/b$ is an integer. The proof follows from applying Theorem~\ref{thm:hd} to the distribution $p$ consisting of $1/b$ domain elements that each occur with probability $b$, and setting $\eps_i = \eps \sqrt{b}.$  Letting $Q$ be the family of distributions defined in Theorem~\ref{thm:hd} associated to $p$ and the $\eps_i$'s, note that with probability $1-o(1)$ it is the case that the first and second conditions in the proposition statement are satisfied.  Additionally, the theorem guarantees that $p$ cannot be distinguished with probability $>2/3$ from such a $q$ given a sample of size $m$ provided that $m < c \left(\sum_i \frac{\eps_i^4}{p_i^2} \right)^{-1/2} = c\frac{\sqrt{b}}{\eps^2}.$

Given an algorithm that could distinguish, with probability at least $3/4>2/3+o(1)$, whether $||p'-q'||_2=0$ versus $||p'-q'||_2 \ge \eps/2$,  using $m=O(\sqrt{b}/{\eps^2})$ samples drawn from each of $p',q'$, one could use it to perform the above (impossible) task of distinguishing with probability greater than $2/3$ whether a set of samples was drawn from $p$, versus a random $q \leftarrow Q$ by running the hypothetical $\ell_2$ tester on the set of samples, and a set drawn from $p$.
\end{proof}

\bigskip

\bibliographystyle{alpha}

\bibliography{allrefs}


\appendix

\section*{Appendix} \label{sec:ap}

\section{An \new{$O(n^{2/3}/\eps^2)$} $\ell_1$-tester} \label{ssec:improved-l1}

In this section, we show how we to obtain an
$\ell_1$ closeness tester with sample complexity  \new{$O(n^{2/3}/\eps^2)$}, by using essentially
the same approach as~\cite{BFR+:00}.

Recall that the $\ell_1$ closeness tester in~\cite{BFR+:00} proceeds in two steps:
In the first step, it ``filters'' the elements of $p$ and $q$ that are ``$b$-heavy'', i.e.,
have probability mass at least $b$ -- for an appropriate value of $b$.
(This step essentially amounts to {\em learning} the heavy
parts of $p$ and $q$.) In the second step, it uses an $\ell_2$ closeness
tester to test closeness of the ``light'' parts of $p$ and $q$. (Note that in the second step
the $\ell_2$ tester needs to be called with error parameter $\eps/\sqrt{n}$.)

Our improvement over~\cite{BFR+:00} is two fold: First, we perform the first step
(learning) in a more efficient way \new{(using a different algorithm)}. Roughly, this improvement
allows us to save a $\log n$ factor in the sample complexity.
Second, we apply our optimal $\ell_2$ tester in the second step.

Regarding the first step, note that
the heavy part of $p$ and $q$ has support size at most $2/b$. \new{Roughly, we show that the heavy part
can be learned to $\ell_1$ error $\eps$ using $O( (1/b)/\eps^2)$ samples (which is the best possible) -- without knowing a priori
which elements are heavy versus light. The basic idea to achieve this is
as follows: rather than inferring {\em all} the heavy elements (which inherently incurs an extra
$\log(1/b)$ factor in sample complexity, due to coupon collector's problem), a
small fraction of heavy elements are allowed to be undetected; this
modification requires a more involved calculation for heavy elements and a
relaxed definition for light elements.

The first step of our $\ell_1$ test uses $s_1 = O( (1/b)/\eps^2)$ samples and the second step
uses $s_2 = O(\sqrt{b}/\tilde{\eps}^2)$ samples, where $\tilde{\eps} = \eps/\sqrt{n}$. The overall sample
complexity is $s_1+s_2$, which is minimized for $b = \Theta(n^{-2/3})$ for a total sample complexity of $O(n^{2/3}/\eps^2)$.
We remark that since the sample complexity of each step is individually optimal, our achieved bound seems to be the best that could possibly be achieved via this reduction-based approach, supporting the view that, in some sense, the more direct approach of Section~\ref{sec:L1} is necessary to achieve the optimal dependence on $\eps$.}

In the following subsections we provide the details of the algorithm and its analysis.

We start with the following definition:
\begin{definition}
  A distribution $p$ is \emph{$(b,C)$-bounded} if $\norm p_2^2 \leq Cb$.
\end{definition}

\subsection{Heavy elements.}
We denote by $\hat p$ (resp. $\hat q$) the empirical distribution obtained after taking $m$ independent samples from $p$ (resp. $q$).
We classify elements into the following subsets:
\begin{itemize}
  \item Observed heavy $H(\hat p) = \{i\mid \hat p_i \geq b\}$ versus observed
    light $L(\hat p) = \{i\mid \hat p_i < b\}$.
  \item Truly heavy $\overline H(p) = \{i \mid p_i \geq b/2\}$ versus truly
    light $\overline L(p) = \{i\mid p_i < b/2\}$.
\end{itemize}
(Note the threshold for the observed distribution is $b$, while for the true distribution is $b/2$.)

Consider the random variables
\begin{align*}
  D_i &= \abs{\hat p_i - \hat q_i} - \abs{p_i - q_i}, \quad
  D(A) = \bigabs{\sum_{i\in A} D_i} .
\end{align*}
We sometimes write $D(AB)$ for $D(A\cap B)$.

We will also use the shorthand $\hat H = H(\hat p)\cup H(\hat q)$.
We want to show that $$\norm{\hat p - \hat q}_{H(\hat p)\cup H(\hat q)} \approx_\eps
\norm{p - q}_{H(\hat p)\cup H(\hat q)}$$ with high probability. To do this, we use the bound
\begin{equation}
  \label{errorbound}
  D(\hat H) \leq D(\hat H\overline H(p)\overline H(q)) + D(\hat H\overline
  H(p)\overline L(q)) + D(\hat H\overline L(p)\overline H(q)) + D(\hat
  H\overline L(p)\overline L(q)) .
\end{equation}

The first three terms on RHS of \eqref{errorbound} will be bounded by
Corollary \ref{heavybound} below. We start with the following simple claim:


\begin{claim}
  \label{diffvar}
  For any $i\in \domain n$,
  \begin{equation}
    \label{singlevariance}
    \E [D_i^2] \leq \frac{p_i+q_i}m.
  \end{equation}
\end{claim}

\begin{proof}
  Expand the LHS of \eqref{singlevariance} as
  \[ \E(\hat p_i-\hat q_i)^2 - 2\abs{p_i-q_i} \E\abs{\hat p_i-\hat q_i} +
    \abs{p_i-q_i}^2. \]
  Since
  \[ \E(\hat p_i - \hat q_i)^2 = \Var[\hat p_i - \hat q_i] + (\E[p_i-q_i])^2 =
    \frac{p_i(1-p_i)+q_i(1-q_i)}m + \abs{p_i-q_i}^2, \]
  the LHS of \eqref{singlevariance} is at most
  \[ \frac{p_i + q_i}m - 2\abs{p_i-q_i}(\E\abs{\hat p_i-\hat q_i} -
    \abs{p_i-q_i}) . \]
  The result follows by the elementary fact $\E |X| \ge |\E X|$ applied to $X = \hat p_i - \hat
  q_i$.
\end{proof}

\begin{corollary}
  \label{heavybound}
  If we use $m \geq 4/(\eps^2 b\delta)$ samples, then for any (possibly random)
  $H\subseteq \overline H(p)$, we have
  \[ D(H) \leq \eps \]
  except with probability $\delta$.
\end{corollary}

\begin{proof}
  By Cauchy--Schwarz,
  \[ D(H)^2 = \left( \sum_{i\in H} D_i \right)^2 \leq \size{\overline H(p)}
    \sum_{i\in \overline H(p)} D_i^2 . \]
  Now we take expectation on both sides.
  Since $\sum_i \E[D_i^2] \leq \sum_i (p_i+q_i)/m \leq 2/m$, and
  $\size{\overline H(p)} \leq 2/b$, we have $\E[D(H)^2] \leq \eps^2\delta$.
  Hence
  \[ \Pr[D(H) \geq \eps] = \Pr[D(H)^2 \geq \eps^2] \leq \delta \]
  by Markov's inequality.
\end{proof}

We bound the last term on the RHS of \eqref{errorbound} by
\begin{equation}
  \label{HLLbound}
  D(\hat HL(p)L(q)) \leq D(H(\hat p)H(\hat q)\overline L(p)\overline L(q)) +
  D(H(\hat p)L(\hat q)\overline L(p)\overline L(q)) + D(L(\hat p)H(\hat
  q)\overline L(p)\overline L(q)) .
\end{equation}
The RHS will be bounded by Corollaries \ref{HHLLbound} and \ref{HLLLbound}
below.

\begin{claim}
  For any $p_i \leq b/2$, any $t \geq 1$, with $m = 1/(\eps^2 b)$ samples,
  \begin{equation}
    \label{light}
    \Pr[\hat p_i \geq t b] \ll \frac{\eps^2}{t^2} \frac{p_i}b .
  \end{equation}
\end{claim}

\begin{proof}
  Note that $\hat p_i$ has distribution $\Bin(m,p_i)/m$, so by Chebyshev's inequality,
  \[ \Pr[\hat p_i \geq tb] \leq \Pr[\abs{\hat p_i - p_i} \geq tb/2] \leq
    \frac{\Var[\hat p_i]}{(tb/2)^2} \leq \frac{4p_i}{m(tb)^2} =
    \frac{4\eps^2}{t^2} \frac{p_i}b .  \qedhere \]
\end{proof}

\begin{lemma}
  \label{lighthp}
  For any $\delta > 0$, for any $\eps \ll \delta$, using $m \gg 1/(\eps^2 b)$
  samples,
  \[ \norm{\hat p}_{\overline L(p)H(\hat p)} \leq \eps \]
  except with probability $\delta$.
\end{lemma}

\begin{proof}
  \begin{align*}
    \E \norm{\hat p}_{\overline L(p)\cap H(\hat p)} &= \sum_{i\in
      \overline L(p)} \E [\hat p_i \cdot \indic{p\geq b}] = \sum_{i\in
      \overline L(p)} \sum_{j\geq 0} \E[\hat p_i \cdot \indic{2^j b \leq p_i
      \leq 2^{j+1}b}] \\
    &\leq \sum_{i\in \overline L(p)}\sum_{j\geq 0} 2^{j+1}b \cdot \Pr[\hat
    p_i\geq 2^jb] \\
    &\stackrel{\eqref{light}}\leq \sum_{i\in \overline L(p)} \frac {C \eps^2
      p_i}b \sum_{j\geq 0} \frac{2^{j+1}b}{2^{2j}} \leq 4C\eps^2 .
  \end{align*}
  By Markov's inequality,
  \[ \Pr\left[ \norm{\hat p}_{\overline L(p)H(\hat p)} \geq \eps \right] \leq
    \frac{4C \eps^2}\eps = 4C \eps \leq \delta. \qedhere \]
\end{proof}

\begin{corollary}
  \label{HHLLbound}
  For any $\delta > 0$, any $\eps \ll \delta$, using $m \gg 1/(\eps^2 b)$
  samples,
  \[ D(H(\hat p)H(\hat q)\overline L(p)\overline L(q)) \leq \eps , \]
  except with probability $\delta$.
\end{corollary}

\begin{proof}
  By triangle inequality,
  \[ \norm{p-q}_{H(\hat p)H(\hat q)\overline L(p)\overline L(q)} \leq
    \norm{\hat p-p}_{\overline L(p)H(\hat p)} + \norm{\hat q-q}_{\overline
      L(q)H(\hat q)} + \norm{\hat p-\hat q}_{H(\hat p)H(\hat q)\overline
      L(p)\overline L(q)} . \]
  The first two terms on the RHS are dominated by $\norm{\hat p}_{\overline
    L(p)H(\hat p)}$ and $\norm{\hat q}_{\overline L(q)H(\hat q)}$.
  By Lemma
  \ref{lighthp},
  \[ \norm{p-q}_{{H(\hat p)H(\hat q)\overline L(p)\overline L(q)}} \leq
    \norm{\hat p-\hat q}_{{H(\hat p)H(\hat q)\overline L(p)\overline L(q)}} +
    \eps \]
  except with probability $\delta/2$.
  We also get the reverse inequality by swapping the roles of $p-q$ and $\hat
  p-\hat q$.
\end{proof}

\begin{corollary}
  \label{HLLLbound}
  For any $\delta > 0$, any $\eps \ll \delta$, using $m \gg 1/(\eps^2 b)$
  samples,
  \[ D(H(\hat p)L(\hat q)\overline L(p)\overline L(q)) \leq \eps \]
  except with probabilty $\delta$.
\end{corollary}

\begin{proof}
  It is easy to see that $\abs{\abs{\hat p_i - \hat q_i} - \abs{p_i - q_i}}
  \leq \hat p_i$ for $i\in H(\hat p)L(\hat q)\overline L(p)\overline L(q)$.
  Hence $$D(H(\hat p)L(\hat q)\overline L(p)\overline L(q)) \leq \norm{\hat
    p_i}_{\overline L(p)H(\hat p)},$$ and the result follows by Lemma
  \ref{lighthp}.
\end{proof}

Applying Corollaries \ref{heavybound}, \ref{HHLLbound} and \ref{HLLLbound} to
inequalities \eqref{errorbound} and \eqref{HLLbound}, we have thus shown the
main theorem of this section.

\begin{theorem}
  \label{heavyestimator}
  For any $\delta > 0$, any $\eps \ll \delta$, using $m \gg 1/(\eps^2b\delta)$
  samples,
  \[ \norm{\hat p - \hat q}_{H(\hat p)\cup H(\hat q)} \approx_\eps \norm{p -
      q}_{H(\hat p)\cup H(\hat q)} \]
  except with probability $\delta$.
\end{theorem}

\subsection{Light elements.}
We now deal with the light elements.
Let $p'$ be the low-frequency distribution constructed in Step~2 of the
$\ell_1$ tester (those elements with empirical frequency at least $b$ have
their weights redistributed evenly).
It will be shown to be $(O(b),O(1))$-bounded in Theorem \ref{typicallylight} below.
\begin{theorem}
  \label{typicallylight}
  $p'$ is $(2b,O(1/\delta))$-bounded except with probability $\delta$.
\end{theorem}

\begin{proof}
  Let $H = \{i\mid p_i \geq 2b\}$ and $\hat L = \{i\mid \hat p_i < b\text{ and
    }\hat q_i < b\}$.
  We wish to bound
  \begin{equation}
    \label{expectnorm}
    \E \left[\sum_{i\in \hat L\cap H} p_i^t\right] = \sum_{i\in H} p_i^t
    \Pr[i\in \hat L]
  \end{equation}
  by $O_t(b^{t-1})$.
  Indeed, writing $p_i = x_ib$, the summand
  \[ p_i^t \Pr[i\in \hat L] \leq p_i^t \Pr[\hat p_i\leq b] = p_i b^{t-1}\cdot
    x_i^{t-1}\bin(m,p_i,<bm). \]
  The factor
  \[ x^{t-1}\bin(m,p_i,<bm) \leq x_i^{t-1} \exp\left( -\frac{Cx_i}{8\eps^2}
    \right) \]
  by a Chernoff bound and equals $O_t(1)$ uniformly in $x_i$ and $\eps$.
  Hence \eqref{expectnorm} is $O_t(b^{t-1})$.
  By Markov's inequality,
  \begin{equation}
    \label{smallheavynorm}
    \sum_{i\in \hat L\cap H} p_i^t \ll_t b^{t-1}/\delta
  \end{equation}
  except with probability $\delta$.

  Note that
  \[ p'_i \leq \left(p_i+\frac 1n\right) \indic{i\in \hat L} + \frac
    1n\indic{i\notin \hat L} , \]
  thus
  \[ \norm{p'}_t^t \leq \sum_{\hat i\in \hat L\cap H} \left(p_i+\frac
      1n\right)^t + \sum_{i\notin H} \left(p_i+\frac 1n\right)^t + \sum_{i
      \notin \hat L}\left(\frac 1n\right)^t . \]
  Together with $(r+s)^t \ll_t r^t + s^t$ and $\sum_i (1/n)^t \leq 1/n^{t-1}
  \leq b^{t-1}$, it follows that $\norm{p'}_t^t \ll_t b^{t-1}/\delta$ whenever
  \eqref{smallheavynorm} holds.
\end{proof}

\begin{theorem}
There exists an algorithm $\ell_1$-Distance-Test that, for $\eps \geq 1/\sqrt n$, uses
  $O(n^{2/3}\eps^{-2})$ samples from $p, q$ and has the following behavior: it rejects with probability
  $2/3$ when $\norm{p-q}_1 \geq \eps$, and accepts with probability
  $2/3$ when $p = q$.
\end{theorem}

\begin{proof}[Proof (Sketch)]
  The algorithm proceeds as follows: We pick $b = n^{-2/3}$.
  We check if the ``$b$-heavy'' parts $H(\hat p)\cup H(\hat q)$ of $p$ and $q$ are
  $\eps/2$-far using Theorem \ref{heavyestimator}.
  We then construct light versions $p'$ and $q'$ as in~\cite{BFR+:00}; these
  distributions are $(b,O(1))$-bounded with high probability by Theorem~\ref{typicallylight}.
  Finally, we check whether they are $\eps/2$-far using Proposition~\ref{prop:ggg} (where we set $\tilde \eps = \eps/\sqrt n$).
  The number of samples we need for both Theorem \ref{heavyestimator} and
  Proposition~\ref{prop:ggg} is $O(n^{2/3}\eps^{-2})$. This completes the proof.
\end{proof}

\end{document}